\documentclass[12pt]{article}
\usepackage{amsmath, amsthm, amsfonts}
\usepackage[margin=1 in]{geometry}
\usepackage{blindtext, xcolor}
\usepackage{algorithm}
\usepackage{algorithmicx}
\usepackage{algpseudocode}
\usepackage{apacite}
\usepackage{natbib}
\usepackage{appendix}
\usepackage{rotating}
\usepackage{hyperref}
\usepackage{float}
\usepackage{graphicx}
\usepackage{multirow}
\usepackage{caption,subcaption}
\usepackage{cleveref}
\usepackage{arydshln}
\usepackage{authblk}
\usepackage{xr,bm}
\usepackage{relsize}

\newcommand{\be} {\begin{eqnarray*}}
\newcommand{\ee} {\end{eqnarray*}}

\theoremstyle{definition}

\newtheorem{theorem}{Theorem}[section]

\newtheorem{proposition}[theorem]{Proposition}

\def\*#1{\bm{#1}}

\title{A Generalized Estimating Equation Approach to Network Regression}
\author[1]{Riddhi Pratim Ghosh}
\author[2]{Jukka-Pekka Onnela}
\author[3]{Ian Barnett}

\affil[1]{Department of Mathematics and Statistics, Bowling Green State University}
\affil[2]{Department of Biostatistics, Harvard University}
\affil[3]{Department of Biostatistics, University of Pennsylvania}

\date{}

\begin{document}
\maketitle
\begin{abstract}
Regression models applied to network data, where node attributes are the dependent variables, pose methodological challenges. As has been well studied, naive regression neither properly accounts for network community structure, nor does it account for the dependent variable acting as both model outcome and covariate. To address this methodological gap, we propose a network regression model motivated by the important observation that controlling for community structure can, when a network is modular, significantly account for meaningful correlation between observations induced by network connections. We propose a generalized estimating equation (GEE) approach to learn model parameters based on node clusters defined through any single-membership community detection algorithm applied to the observed network. We provide a necessary condition on the network size and edge formation probabilities to establish the asymptotic normality of the model parameters under the assumption that the graph structure is a stochastic block model. We evaluate the performance of our approach through simulations and apply it to estimate the impact of the county-level commercial airline transportation network on COVID-19 incidence rates and on net financial aid given or received.

\end{abstract}

\noindent%
{\it Keywords:} network regression; transportation networks; generalized estimating equations; COVID-19
\vfill

\newpage

\section{Introduction}

For nodes of a network, some attributes may be influenced by the same attributes of their network neighbors. Network regression aims to model node attributes through regression while accounting for the role of these neighboring nodes.
Two important methodological challenges of network regression are (1) accounting for correlation induced by network structure, and (2) allowing for node attributes to appear as both the dependent and independent variables for different observations in the same model.  For example, network regression was applied to study the person-to-person spread of obesity where \cite{christakis2007spread} estimated the degree to which obesity is spread through social ties as opposed to being explained by homophily, the phenomenom often typified by the saying ``birds of a feather flock together'' \citep{shrum1988friendship, igarashi2005gender, tifferet2019gender}. Sensitivity analyses suggest that contagion effects for obesity and smoking cessation are reasonably robust to possible latent homophily or environmental confounding while those for happiness and loneliness are somewhat less so \citep{vanderweele2011sensitivity}. To further investigate the causal relationship, \cite{shalizi2011homophily} provided three factors underlying such interaction: homophily, or the formation of social ties due to matching individual traits; social contagion; the causal effect of an individual's covariates on his or her measurable responses; and an individual's response. This has led to the development of a more general framework of network regression that models observed individual node attributes with covariates and the network interactions \citep{hoff2002latent, hoff2021additive}. Here we will show how community structure can be exploited to solve, or at least alleviate, the above stated methodological challenges. This introduces an additional challenge and source of variability: estimation of community membership. Some community detection methods have been developed to incorporate covariate information \citep{binkiewicz2017covariate, mu2022spectral}.  While many existing network regression approaches assume community labels are known in order to simplify the estimation of covariate effects, this assumption seldom holds in practice. In more realistic settings, community labels are unobserved with differences in covariates existing throughout the communities, and this makes model inference challenging.  

There are many cases of network regression that assume community membership labels to be known. For \cite{aukett1988gender} showed that gender difference plays an important role in friendship patterns: women shows a preference for a few, closer, intimate same-sex friendships based on sharing emotions whereas men build up friendship based on the activities they do together. In another case, \cite{staber1993friends} studied how women and men form entrepreneurial relationships, concluding that women's networks are wider with more strangers and higher proportion of cross-sex ties. In each of these network regression examples, community labels are assumed known which alleviates the inferential task of finding the effects of covariates and learning community structures. However, in other scenarios, community membership labels cannot be assumed. For example, we will consider the network regression problem of seeing the impact of air travel network flow between countries on COVID-19 incidence rates and net foreign aid, in which case the community structure of the commercial air transit networks must be estimated.

Generalized estimating equations (GEE) \citep{liang1986longitudinal, zeger1986analysis} provide a popular approach that is often used to analyze longitudinal and other types of correlated data \citep{burton1998extending, diggle2002analysis, mandel2021neural}. \cite{zhou2016regression} proposed hybrid quadratic inference method exploiting both prior and data-driven correlations among network nodes for assessing relationships between multi-dimensional response variables and covariates that are correlated within a network. Here, we propose first performing community detection in order to estimate community membership, and then using a GEE approach to account for the  resulting estimated community structure. This approach is general and agnostic to the community detection algorithm used so long as each node belongs to at most one community. (This excludes so-called soft partitioning methods, which may allow a node to belong to multiple communities.) Given that network regression needs to account for correlation between attributes from nodes that are connected, GEEs allows for arbitrary correlation between nodes within the same community. Of note, this approach is best suited for highly modular networks so that the GEE assumption of independence between different communities is more accurate, though we explore its performance in less modular networks with more between-community mixing.

The rest of this article is organized as follows. In \Cref{Methods} we introduce our network regression model along with our GEE extension, and we describe the transportation network we use to model country-specific COVID-19 incidence rates and the financial aid received by countries. In \Cref{Results} we present the theoretical results followed by extensive simulation results and real data analysis. Finally, \Cref{Discussion} concludes the article with a  discussion.

\section{Methods}\label{Methods}

In this section, we introduce our network regression model and describe the air travel networks for each month of the first quarter of 2020 (the beginning of the COVID-19 pandemic) with country-specific information such as COVID-19 incidence rate, financial aid received or given, GDP, and population size. Next, we present a generalized estimating equation (GEE) approach to network regression that accounts for network-induced correlation between observations.

\subsection{Model and notation}\label{Notations and model}

We consider a directed network of $n$ nodes and an $n \times n$ adjacency matrix $\bm A=(a_{ij})$ with $0$'s on the diagonal (i.e. no self-loops).  
 We denote the feature variable of the  $i$th node by $y_i$, the $l$-dimensional vector of covariates by $\bm{\alpha}$, and the corresponding design matrix by $\bm{x}_i$. Denoting the coefficient of interest by $\beta$, the network regression model for node attribute $y_i$ we consider the generalized linear model framework:
 \begin{align}\label{original_eqn1}
& g(\mu_i)=\bm{\alpha}^{\top}\bm{x}_i + \beta \cdot \sum_{j \ne i}A_{ji}g^{-1}(\mu_j)/(n-1) ,
\end{align}
where $\mu_i=E(y_i|\bm x_i, \bm A)$, $\bm{\alpha} \in \mathbb{R}^l$, $g(\cdot)$ is the link function, and $g^{-1}(\cdot)$ its inverse. Next we consider the continuous node attribute and binary node attribute cases.

\textit{(i) Regression with continuous outcome:} With the identity link function, this becomes
\begin{align}\label{cont_outcome_eqn}
&\mu_i=\bm{\alpha}^{\top}\bm{x}_i + \beta \cdot \sum_{j \ne i}A_{ji}\mu_j/(n-1).
\end{align}
A slightly alternative approach is to directly model the $y_j$ themselves in place of $\mu_j$.
\begin{align}\label{ar_model}
& y_i=\bm{\alpha}^{\top}\bm{x}_i + \beta \cdot \sum_{j \ne i}A_{ji}y_j/(n-1)+\epsilon_i,
\end{align}
where $\epsilon_i$s iid zero-mean normal error terms.
One can note that this  model is reminiscent of the first-order autoregressive spatial model of \cite{kelejian1998generalized} which frequently contains a spatial lag of the dependent variable as a covariate that is spatially autoregressive. In contrast to the naive regression model, one difficulty with the systematic component of \eqref{ar_model} is that it now contains $y_j$, thus forcing it to play the role of both the independent and  dependent variable. 

\textit{(ii) Regression with binary outcome:} For the binary response variable $y_i$ taking values in $\{0,1\}$, one can use logit link function which gives the following model:

\begin{align}\label{binary_outcome_eqn}
\text{logit}(\mu_i)=\bm{\alpha}^{\top}\bm{x}_i + \beta \cdot \sum_{j \ne i}A_{ji}\text{logit}^{-1}(\mu_j)/(n-1),
\end{align}
where $\text{logit}(\mu_i)=\text{log}(\mu_i/(1-\mu_i))$, $\text{logit}^{-1}(\mu_j)=\text{exp}(\mu_j)/(1+\text{exp}(\mu_j)$ and  $\mu_i=P(y_i=1)$.

Using vector and matrix notation, the generic model in Equation \eqref{original_eqn1} can also be written as
\begin{align*}
  & g(\bm \mu) = \bm X^{\top} \bm \alpha + \beta \cdot \bm Ag^{-1}(\bm \mu)/(n-1),
\end{align*}
where $g(\bm \mu)=[g(\mu_1),g(\mu_2),\cdots, g(\mu_n)]^{\top}$, $g^{-1}(\bm \mu)=[g^{-1}(\mu_1),g^{-1}(\mu_2),\cdots, g^{-1}(\mu_n)]^{\top}$ are the concatenated vectors of the link function and its inverse, and $\bm X=[\bm x_1: \bm x_2: ...: \bm x_n]$, is the matrix of covariates. 


\subsection{Data description}\label{Data description}

With 622 million confirmed cases and 6.5 million deaths globally as of October 01, 2022, the COVID-19 pandemic has had a tremendous impact on the world, shrinking the global economy by 5.2\%, the largest recession in the history post World War II \citep{world2020covid}. The travel bans in places worldwide have severely affected the tourism industry, with estimated losses of 900 billion to 1.2 trillion USD and tourism down 58\%-78\% \citep{le2022framework}. The airline industry has also suffered heavily, with 43 airlines declaring bankruptcy and 193 of 740 European airlines at risk of closing. Here we focus on the start of the pandemic covering the transition to travel bans across the world to study the relative effectiveness of travel bans for controlling and contributing to COVID-19 incidence rates. 

We use pandemic data from the Johns Hopkins University coronavirus data repository through April 30, 2020, \citep{CSSEGISandData}. Flight data are from the Official Airline Guide (OAG) \citep{strohmeier2021crowdsourced}. Because only data for January and February 2020 are available from OAG, we used the estimated fight data for other time periods using the OpenSky Network database \citep{schafer2014bringing, strohmeier2021crowdsourced}. This database tracks the number of fights from one country to another over time, which we use to estimate country-to-country flight data for other months. We include as covariates the GDP, total population, and percentage of the urban population for each country in our network.  Constructing a network based on by the flight data and incorporating the above country-specific attributes such as GDP, population, etc. as covariates through $\bm \alpha$ in  \eqref{cont_outcome_eqn}, we aim to estimate the effectiveness of travel bans through $\beta$ in model \eqref{cont_outcome_eqn}. In a subsequent analysis of binary response, we investigate the impact of network data from January 2020, along with other covariates, on financial aid using the model described in Equation \eqref{binary_outcome_eqn}. The financial aid data from 2019, representing financial aid received or given by countries in US dollars (million), has been dichotomized to 0 and 1 after thresholding.


\subsection{Generalized estimation equation (GEE) approach}\label{GEE approach}
Our network contains $K$ communities where community $k$ is defined by 
$$E_k=\{i: g_i=k\},$$
where index $g_i$ represents the community membership  of node $i$, $|E_k|=n_k$ (number of nodes in community $k$) and $\sum n_k=n$. Let $\bm y_k$ denote the concatenated vector of $y_{ij}$s, and $\bm X_k=[\bm x_1, \bm x_2, ...,\bm x_{n_k}]$ denote the $l\times n_k$ sub-matrix of covariates corresponding to cluster $k$.

To fit our network regression model in the GEE framework, we use network communities as clusters in the GEE, and we use the following equation to model the node attributes of members of the $k$th cluster:  
\begin{align}\label{network_reg_gee}
&  g(\bm \mu_k) = \bm X_k^{\top}\bm \alpha + \beta \bm Z_k ,  
\end{align}
where  $\bm A_k$ is the $n_k \times n_k$ sub-matrix of $\bm A$  pertaining to the cluster $k$, $\bm \mu_k$ is the marginal mean of the response pertaining to the cluster $k$, i.e., $\bm \mu_k=E(\bm y_k\vert \bm X_k,\bm A_k)$, and $\bm Z_k=\bm A_k g^{-1}(\bm \mu_k) /(n-1)$. For the model in \eqref{ar_model} the explicit formula for $\bm \mu_k$ is:
$$\bm \mu_k=E(\bm y_k|\bm X_k,\bm A_k)=(\bm I_{n_k}-\beta \bm A_k /(n-1))^{-1}\bm X_k^{\top}\bm \alpha, \;\;\;\ k=1,2,...,K,$$
where $\bm I_{n_k}$ is the identity matrix of order $n_k$.

Adopting a GEE approach \citep{liang1986longitudinal}, the resulting estimating equation is given by
\begin{align}\label{gee_objective_fn}
& \sum_{k=1}^{K}\bm D_k^{\top} \bm V_k^{-1}(\bm y_k-\bm \mu_k)=\bm 0, \;\;\ k=1,2,...,K,
\end{align}
where $\bm D_k=\frac{\partial \bm \mu_k}{\partial (\beta,\bm \alpha)^{\top}}$ is of dimension $n_k \times (l+1)$ and $\bm V_k$ is the $n_k\times n_k$ working covariance matrix of $\bm y_k$. The form for $\bm D_k$ using chain rule is:
\begin{align*}
\bm D_k = \underbrace{\frac{\partial \bm\mu_k}{\partial g(\bm \mu_k)}}_{n_k\times n_k} \left[\underbrace{\frac{1}{n-1}\bm A_k\frac{\partial g^{-1}(\bm\mu_k)}{\partial \beta}}_{n_k\times 1},  \underbrace{\bm X_k^{\top}+\frac{\beta}{n-1}\bm A_k\frac{\partial g^{-1}(\bm \mu_k)}{\partial \bm\alpha}}_{n_k\times l}\right].
\end{align*}

For the regression model specified in Equation \eqref{ar_model}, the explicit formula for $\bm D_k$ is:

$$\bm D_k=\left[-\underbrace{(\bm I_{n_k}-(\beta/(n-1)) \bm A_k)^{-1}\bm A_k(\bm I_{n_k}-(\beta/(n-1)) \bm A_k)^{-1}\bm X^{\top}\bm \alpha}_{n_k\times 1} ,  \underbrace{(\bm I_{n_k}-(\beta/(n-1)) \bm A_k)^{-1}\bm X_k^{\top}}_{n_k \times l}\right].$$
One can note that $\bm D_k$ consists of two partitioned matrices where the first one corresponds to the network parameter $\beta$ and the second one is due to the covariate $\bm \alpha$. We can solve for $\hat{\bm \alpha}$ and $\hat{ \beta}$ in equation \eqref{gee_objective_fn} through iterative reweighted least squares, and use the robust sandwich covariance estimator to perform inference on $\bm \alpha$ and $ \beta$.
\section{Results}\label{Results}

\subsection{Theoretical results}\label{Theoretical results}

In this section, we prove the asymptotic normality of the resulting GEE estimator for $\beta$ and $\bm \alpha$ jointly. Towards this goal, we assume constant probabilities of edge formation between and within communities,  where we denote these quantities by $p$ and $q$, respectively, as in a stochastic block model \citep{holland1983stochastic}. Our proof of asymptotic normality hinges on Theorem 2 of  \cite{liang1986longitudinal} which establishes the asymptotic normality of the regression parameter in the classical GEE approach under the assumption  that the  correlation parameter appropriately scaled by the number of communitites is consistent. Our primary distinction from this approach is that we must account for probabilities of edge formation instead of correlation between observations. Therefore, we first establish the consistency of $p$ and $q$ following Proposition 1 of \cite{chen2021analysis} and subsequently show asymptotic normality of $(\hat{\beta},
\hat{\bm \alpha})$ from the estimating equation \eqref{gee_objective_fn}.

\subsubsection{Consistency of $\hat{p}$ and $\hat{q}$}\label{Consistency results}

\begin{proposition}\label{const_p_q}
Consider a  network generated from a stochastic block model of $K$ communities of size $m$ so that total number of nodes is $n=Km$.  Assume that $Km^{\gamma}p \rightarrow p^{*}$ and $K^2m^{\gamma}q \rightarrow q^{*}$ as $m \rightarrow \infty, K \rightarrow \infty$, where $p^{*}$ and $q^{*}$ are positive fixed constants, and $\gamma \in [0,2)$, then $m^{1+\gamma/2}K^{1/2}(\hat{p}-p)$ and $m^{1+\gamma/2}K(\hat{q}-q)$ are both $o_P(1)$, where $\hat{p}$ and $\hat{q}$ are the empirical estimates of $p$ and $q$ from the adjacency matrix following a community detection algorithm. 
\end{proposition}

\begin{proof}
See the appendix in \Cref{proof_const_p_q}.
\end{proof}

Guided by the above proposition, we establish the asymptotic normality of the model parameters in the following theorem. The key difference with the classical formula is the inclusion of the community size in the covariance coming from the consistency of $p$ and $q$. Two remarks are in order.

\textit{Remark 1:} It is instructive to note the inclusion of the number of communities $K$ in the above expression- for consistency of $\hat{p}$ the order is $K^{1/2}$ while for $\hat{q}$ the order is $K$. In other words, this implies that the estimate of between community edge probability benefits from increase number of communities more than the estimate of within community edge probability.

\textit{Remark 2:} While doing asymptotics on the network size $n$, i.e. $n \rightarrow \infty$, we impose the condition that both community size $m$, and the number of communities $K$ tend to infinity. This rules out the existence of a large community can diminish the network topology of other local communities.

\subsubsection{Asymptotic normality of $\hat{\beta}$}

\begin{theorem}\label{asym_norm}
Under the conditions of \Cref{const_p_q}  $K^{1/2}m^{1+\gamma/2}\big((\hat{\beta},\hat{\bm \alpha})-(\beta,\bm \alpha)\big)^{\top}$ is asymptotically multivariate normal with zero mean and covariance given by
$$V=\underset{K\rightarrow \infty}{\text{lim}}K m^{2+\gamma}\Big(\sum_{k=1}^{K}\bm D^{\top}_k\bm V_k^{-1}\bm D_k\Big)^{-1}\Big(\sum_{k=1}^{K}\bm D^{\top}_k\bm V_k^{-1}cov(\bm Y_k)\bm V_k^{-1}\bm D_k\Big)\Big(\sum_{k=1}^{K}\bm D^{\top}_k\bm V_k^{-1}\bm D_k\Big)^{-1}.$$
\end{theorem}

\begin{proof}
See the appendix in \Cref{asym_norm_pf}. 
\end{proof}
\textit{Remark 3:} The variance formula involves the term $m^{2+\gamma}$, where $m$ is the community size. If $m$ is large, then the covariance will increase at a rate $m^{2+\gamma}$. This is reminiscent
of the well-known fact that sandwich estimator $\hat{V}$ of the covariance of $(\beta,\bm \alpha)^{\top}$ is not stable if $m$ is large relative to $K$. In essence, if $m$ grows at a similar rate to $K$ the sandwich estimator becomes and unstable estimator of covariance. In practice this implies that the GEE approach works best when the network contains many smaller communities rather than only a few larger communities.


\subsection{Simulations}\label{Simulations}
We simulate a network of $n$ nodes having balanced communities of size $m=10$ via a stochastic block model and vary $n$ in $\{200, 400\}$ with the number of communities ($K$) being $20$ and $40$ for both values of $n$, respectively.  Let $p$ and $q$ denote the within community and between community probabilities of an edge formation as in a stochastic block model. In each setting, we vary $(p, q)$ across the following values: $(0.8, 0)$, $(0.7, 0.1)$, $(0.6, 0.2)$, and $(0.5, 0.3)$ representing a range from stronger to weaker community structure.

\subsubsection{Estimation of $\beta$ and $\alpha$} 
The choice of true model parameters and the corresponding data-generating process are detailed to span networks with varying degrees of modularity. We set  $\beta_0=0.5$ and $\bm \alpha_0=(1, 1, 1, 1, 0.5, 0.5, 0.5, -0.5, -0.5, 2)^{\top}$ ($l=10$ in equation \eqref{original_eqn1}).  We simulate the $l \times n$ design matrix $\bm X$ by a multivariate normal distribution in the following manner. The $j$th column of $\bm X$, if it belongs to community $k$ ($k=1,2,...,K$), follows MVN$((k/10) \bm 1_l, 0.01 \bm I_l)$, where $\bm 1_l$ and $\bm I_l$ are the vector with 1s and identity matrix of dimension $l$, respectively. In each setting, the adjacency matrix $\bm A$ of dimension $n$ is simulated by the stochastic block model which has $K$ communities of the aforementioned size such that within and between community edge probabilities are $p$ and $q$, respectively. Finally, the response variable $y_i$ is generated as follows:
\begin{enumerate}
\item[] (i) Continuous outcome: We use Equation \eqref{cont_outcome_eqn} to generate $\mu_i$ first, and then $y_i$ is generated from a normal distribution with mean $\mu_i$ and variance 0.01.
\item[] (ii) Binary outcome: First, Equation \eqref{binary_outcome_eqn} is used to generate $\text{logit}(\mu_i)$, and then $y_i$ is generated from a Bernoulli distribution with parameter $\mu_i$.
\end{enumerate}

To estimate $\beta$ and $\bm \alpha$, we first perform a community detection on the graph obtained from the adjacency matrix $\bm A$ as in \citet{rosvall2007maps}. Next, with the resulting communities we fit a GEE using the \textit{geepack} R package \citep{geeglm} and report the estimates of bias and variance by taking the average of over $B=1000$ replications.

The squared bias and variance of the estimated $\beta$ increase as the networks become less modular (i.e., as $q$ increases) for both our GEE approach and naive least squares for continuous outcomes (see \Cref{cont_plot}), as well as for our GEE approach and generalized linear model (GLM) (see \Cref{binary_plot}). The naive least squares and GLM methods do not assume any community structure and report the parameter estimates by assuming independence between observations, using equation \eqref{original_eqn1} directly. Our simulation results confirm that the parameter estimates of our GEE approach coincide with the least squares method or GLM as expected. The differnce comes in the values of the standard errors. 

From \Cref{cont_plot} and \Cref{binary_plot} we see that less modular networks are more difficult to fit in general based on the increase in bias for both methods. Despite this, our GEE approach exhibits less bias for less modular network than naive least squares or GEE, demonstrating that controlling for correlation within communities is effectively accomplished by GEE. This also exposes the weakness of our GEE approach: if the network is not modular with high degrees of mixing between communities (i.e. large $q$) the GEE framework cannot accommodate this due its assumption of independence between communities. However, even when $q$ is large we still observe that GEE somewhat mitigates the impact of correlation induced by the network at least partially, which explains why its bias is smaller than that of naive least squares or GLM. Essentially, even when the network structure suggests the GEE assumptions are incorrect, one is still generally better off partially accounting for network structure using our GEE approach rather than ignore community structure altogether.

The smaller standard errors demonstrated by least squares in \Cref{cont_plot} and by GLM in \Cref{binary_plot} reflect the inaccuracy of the method. By ignoring community structure and assuming independence between observations in network regression settings, we expect to see anti-conservative hypothesis tests, too-narrow confidence intervals, and general overconfidence. This overconfidence is reflected in the too-small standard errors for naive least squares as well as in the anti-conservative Type I errors we demonstrate next.

\begin{figure}[h]
\centering
\includegraphics[width=185mm]{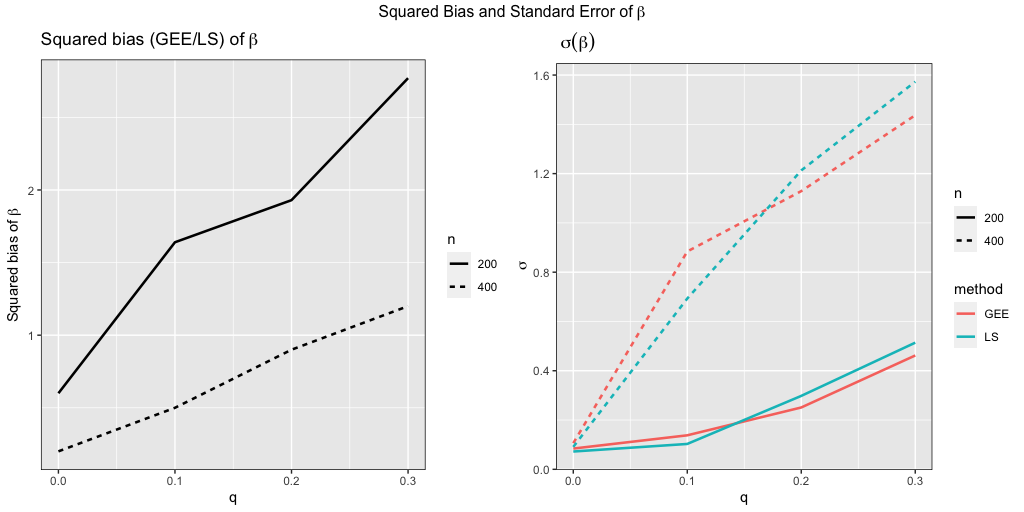}
\caption{\textbf{Bias squared and standard error of estimates of $\beta$ with varying degrees of network modularity for continuous outcomes.} In order to estimate $\beta$, $B=1000$ replications were performed for $n=200, 400$ with the average degree 7 and 16 respectively. The values of $(p, q)$ are varied over $\{(.8,0), (.7, .1), (.6,.2), (.5,.3)\}$. Biases are equal for the GEE and LS methods since they yield the same parameter estimates regardless of the network size $n$.} 
\label{cont_plot}
\end{figure}

\begin{figure}[h]
\centering
\includegraphics[width=185mm]{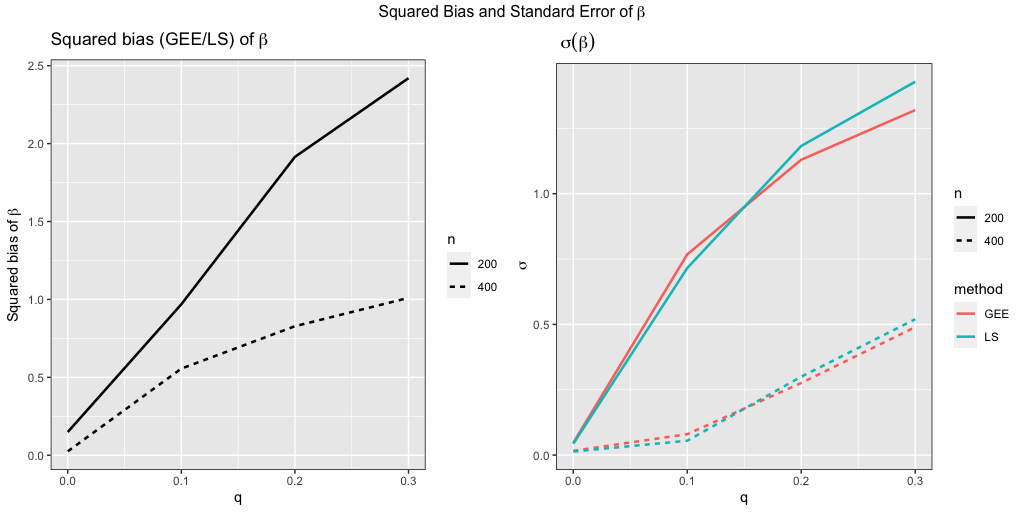}
\caption{\textbf{Bias squared and standard error of estimates of $\beta$ with varying degrees of network modularity for discrete outcomes.} In order to estimate $\beta$, $B=1000$ replications were performed for $n=200, 400$ with the average degree 7 and 16 respectively. The values of $(p, q)$ are varied over $\{(.8,0), (.7, .1), (.6,.2), (.5,.3)\}$. Biases are equal for the GEE and LS methods since they yield the same parameter estimates regardless of the network size $n$.} 
\label{binary_plot}
\end{figure}
\subsubsection{Hypothesis testing of $\beta$}
We consider an approach to testing the hypothesis of $H_0:\beta=0$ against the alternative $H_A:\beta\ne 0$. We perform a simulation study to obtain Type I error in a variety of network structures. We consider two working correlation structures for our GEE model: independence and exchangeable. 
First, we simulate our data as in \Cref{Simulations} with $\beta=0$. We obtain an empirical null distribution by replicating this procedure $B=1000$ times to obtain $\hat{\beta}^{(1)}, \hat{\beta}^{(2)},..., \hat{\beta}^{(B)}$. We estimate the P-value by $\sum_{b=1}^{B}I(\vert \hat{\beta}^{(b)}\vert<\hat{\beta})/B$.  We summarize our Type I error results at the $0.05$ significance level in \Cref{tab0} which demonstrates that as network modularity decreases ($q$ gets larger), the hypothesis test for $H_0$ becomes anti-conservative. This is true for both least squares as well as GEE, although least squares is more adversely impacted. In addition, while GEE excels in the highly modular setting, as expected, LS is still anti-conservative. This demonstrates the efficacy of GEE as a better inferential tool over naive least squares in network regression settings, even when GEE assumptions of between-community independence do not hold. It is also instructive to note that although the P-values for independent correlation structure are generally less than those for exchangeable structure, the differences are not overwhelming.

\begin{table}[h]

\caption{ \textbf{ Type I error for the test of $H_0: \beta=0$ for continuous outcomes} Comparison of Type-I error at the $0.05$ significance level for a network of $n$ nodes with $K$ balanced communities for different choices of within community edge probability ($p$) and between community edge probability ($q$) among GEE with independent
and exchangeable correlation structure and naive least squares.}\label{tab0}
\begin{center}
\begin{tabular}{ c | c |c |c|c} 
 \hline
$(n,K)$ & $(p,q)$ & GEE (independent)& GEE (exchangeable) & LS \\ 
 \hline
 &(0.8, 0) & 0.049 & 0.055 &0.062  \\ 
$(200,20)$ &(0.7, 0.1) & 0.050 & 0.055 & 0.079 \\
 &(0.6, 0.2) & 0.058 & 0.061& 0.091 \\
 &(0.5, 0.3) & 0.072 & 0.075 & 0.095 \\
 \hline 
 &(0.8, 0) & 0.050 &0.050 &0.060  \\ 
$(400,40)$ &(0.7, 0.1) & 0.051 & 0.056& 0.075 \\
 &(0.6, 0.2) & 0.059 & 0.062 & 0.090 \\
 &(0.5, 0.3) & 0.070 & 0.075 & 0.096 \\
\hline
\end{tabular}%
\end{center}
\end{table}

\newpage

\begin{table}[h]

\caption{ \textbf{ Type I error for the test of $H_0: \beta=0$ for binary outcomes} Comparison of Type-I error at the $0.05$ significance level for a network of $n$ nodes with $K$ balanced communities for different choices of within community edge probability ($p$) and between community edge probability ($q$) among GEE with independent
and exchangeable correlation structure and GLM.}\label{tab0.1}
\begin{center}
\begin{tabular}{ c | c |c |c|c} 
 \hline
$(n,K)$ & $(p,q)$ & GEE (independent)& GEE (exchangeable) & GLM \\ 
 \hline
 &(0.8, 0) & 0.045 & 0.060 &0.049  \\ 
$(200,20)$ &(0.7, 0.1) & 0.055 & 0.055 & 0.056 \\
 &(0.6, 0.2) & 0.060 & 0.070& 0.062 \\
 &(0.5, 0.3) & 0.063 & 0.075 & 0.065 \\
 \hline 
 &(0.8, 0) & 0.050 &0.055 &0.050  \\ 
$(400,40)$ &(0.7, 0.1) & 0.060 & 0.060& 0.061 \\
 &(0.6, 0.2) & 0.059 & 0.065 & 0.063 \\
 &(0.5, 0.3) & 0.065 & 0.078 & 0.068 \\
\hline
\end{tabular}%
\end{center}
\end{table}

\subsection{Application to commercial air traffic networks}\label{Real data analysis}

Air travel networks were constructed using flight data obtained from the Official Airline Guide (OAG). The node attribute outcomes consisted of country-specific COVID-19 incidence rates by month, available from the Johns Hopkins University coronavirus data repository through April 30, 2020, \citep{CSSEGISandData} for continuous outcomes. For binary outcomes, the node attribute was financial aid in millions of US dollars for the year 2019 \citep{worldbank_aid_2022}, dichotomized and thresholded by the median. In addition we included country-specific GDP \citep{worldbank2022_gdp}, 
total population \citep{worldbank_population_2022},
and percentage of the urban population \citep{worldbank_urbanpop_2022}
of the countries from the website of the World Bank as covariates. We study (i) to assess the effectiveness of travel bans on the incidence rate of COVID-19 using our network regression model for continuous outcome in Equation \eqref{cont_outcome_eqn}, conducting a month-by-month analysis from January to April 2020, spanning the period from before the pandemic to one month after the travel restrictions were implemented, and (ii) to examine the impact of network structures on the financial aid received or provided by countries in the year 2019, using commercial flight network data from January 2020 using Equation \eqref{binary_outcome_eqn}. In addition, we study the importance of baseline covariates effects such as GDP, percentage of urban population, and population size of the countries vs the network effect. The GDP and the population of the most populated countries are in the order of $10^{12}$ and $10^6$ respectively, so we scale these variables by $10^{12}, 10^6$, and $10^2$, respectively, in order to stablize coefficient estimation.

With these data retrieved across different continents to model the network effect through our proposed model in \eqref{original_eqn1} we assume a directed stochastic block model framework, where nodes correspond to countries and each block contains countries having a large number of commercial flights traveling among them compared to the others. Edge formation is determined by thresholding the population-normalized count of flights arriving at the destination country. 

\begin{figure}
\centering
\includegraphics[width=155mm]{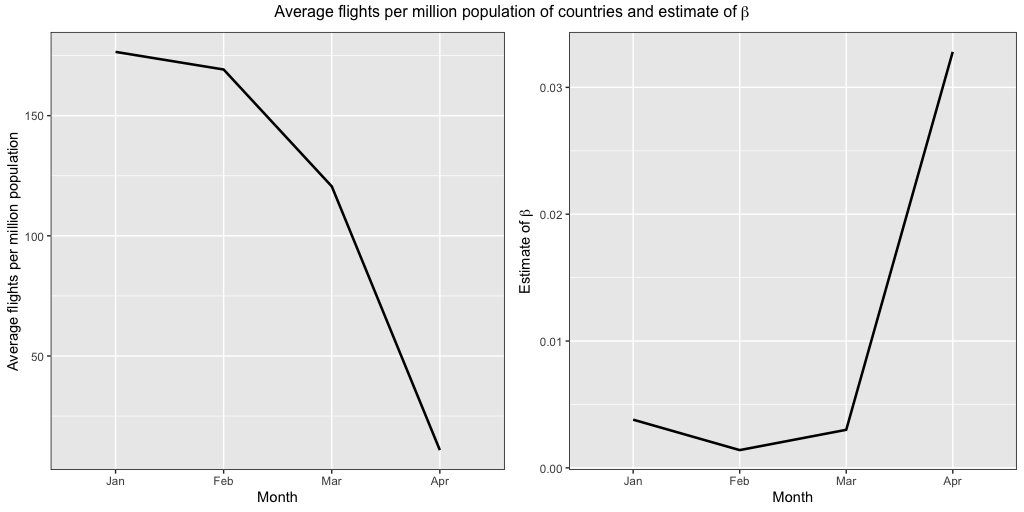}
\caption{\textbf{Changes in air travel networks and the impact of incidence rates over time.} Monthly average number of flights per million population over the countries and monthly estimate of $\beta$ for Jan-Apr 2020 demonstrate a downward and an upward trend, respectively.} \label{monthly_plot}
\end{figure}

\begin{figure}
\centering
\includegraphics[width=155mm]{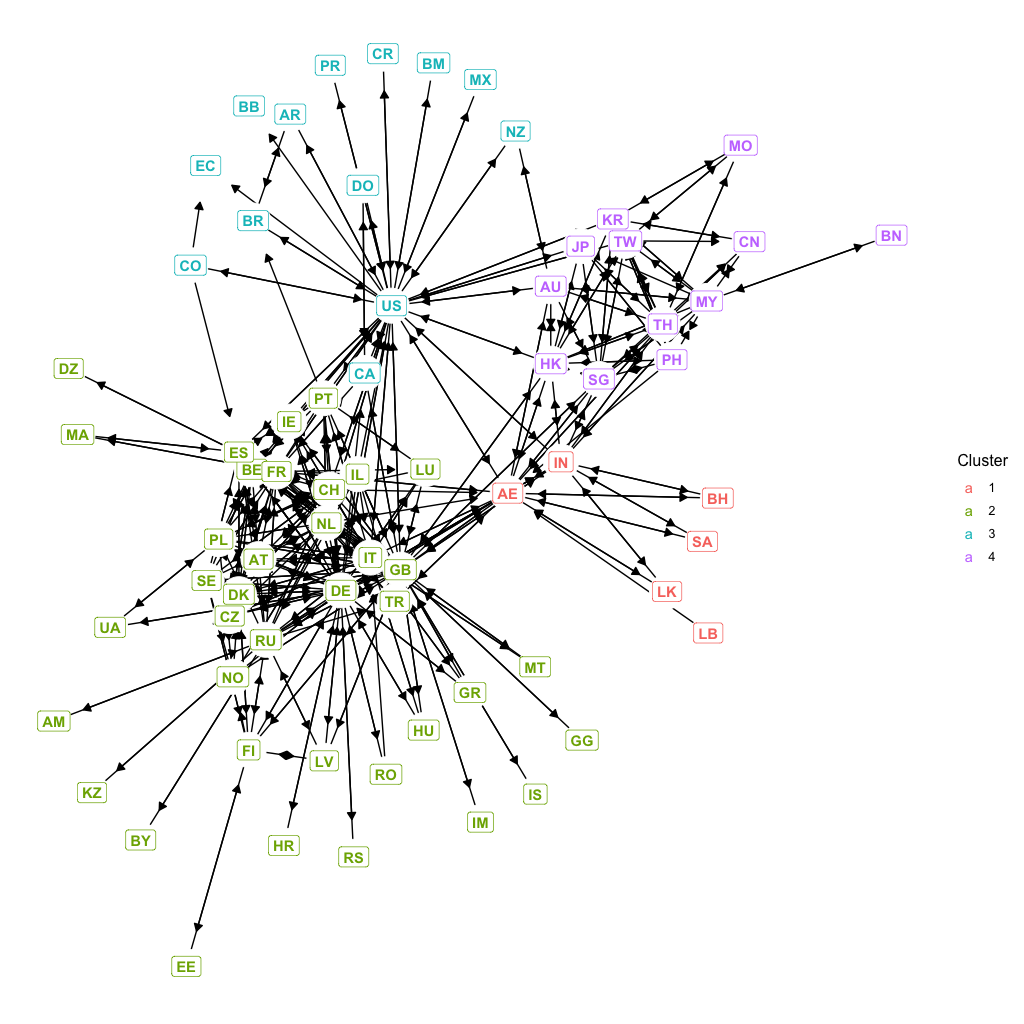}
\caption{\textbf{Country-to-country air traffic network.} The directed graph is obtained from the adjacency matrix of air traffic flow in January 2020, which has been normalized by the country's population per million and thresholded by the third quartile. Nodes with country initials (alpha-2 codes) represent the corresponding countries.} \label{cluster_view}
\end{figure}

We define the variables as follows: (i) $y_i$ represents the incidence rate in the $i$th country, measured as the number of COVID-19 cases per 1000 population. (ii) $y_i \in \{0,1\}$ is an indicator variable, denoting whether the country received financial aid in 2019, based on a median thresholding of the financial aid data. The covariates are represented by the vector $\bm x_{i}$, where $\bm x_{i}$ is of dimensions $4 \times 1$. The vector includes an intercept term, denoted by $1$, and three additional elements: $ x_{i2}$, $x_{i3}$, and $ x_{i4}$. These additional terms contain the population size, GDP, and percentage of urban population for the $i$th country, scaled by $10^6$, $10^{12}$, and $10^2$ respectively.  For constructing the adjacency matrix $\bm A$, we consider two scenarios: unweighted and weighted. For unweighted adjacency matrices, from the directed graph dictated by the number of flights in a particular month, we first construct a count matrix $\bm C$ that counts the number of flights from one country to the other. Then we construct an unweighted adjacency matrix $\bm A$ with 0s and 1s where the entry of $\bm A$ is 1 if it exceeds the third quartile of the elements of $\bm C$, and 0 otherwise. For weighted adjacency matrices we divide the elements of $\bm C$ by the total population of the destination country per million. The main rationale behind this scaling is that one would expect more flights traveling to a populated country compared to a less populated country. Therefore, dividing it by the population of the country will result in the elements of the weighted adjacency matrix being on the same scale. 
 
With the two adjacency matrices constructed in this way, we fit model \eqref{gee_objective_fn} and summarize the results in \Cref{tab1} and \Cref{tab:se_cont} for COVID-19 data and in \Cref{aid_fat_table} for the financial aid data. The estimates of $\beta$ in \Cref{tab1} increase from February to April both for the weighted and unweighed cases suggesting that there is an increasing association between the travel and spread of the pandemic. To further investigate this behavior, we plot the average number of flights per million population in \Cref{monthly_plot} which shows that although the average number of flights is decreasing from January to April, the estimates of $\beta$ are increasing. This reflects the fact that while travel bans led to a decrease in the total number of flights in March and April, the increasing $\beta$ in these months implies that each flight had an increased likelihood of transmitting COVID-19 to the destination country, increasing the correlation between the incidence rates in the two countries.

\Cref{tab1} also demonstrates that while the overall population of a country has a negligible effect on the incidence rate leading to smaller values of $\alpha_2$, the percentage of the urban population plays a crucial role especially when the travel ban has already been in effect, for example, the value of $\alpha_4$ is 62 (55 for weighted case) in April compared to 0 in February. Moreover, the corresponding values are consistent for weighted and unweighted networks.

\Cref{tab1} further demonstrates the utility of including baseline covariates alongside network effect to study the effectiveness of travel bans in mitigating the spread of COVID-19. By comparing the values of $\beta$ and $\alpha_4$, one can note that for the initial months of the pandemic, the network effect is more compared to the effect of urban population. However, when the travel ban has already taken place during March (for most of the countries), the effect of urban population supersedes the network effect as its value increase to 62.33 drastically from 9.88 suggesting that urbanity is the next important factor to consider for controlling the spread of the pandemic after the travel bans. For the financial aid data, \Cref{aid_fat_table} demonstrates that the transportation network has a significant effect on the financial aid received or given by countries compared to the other covariates involved in our study.

Next, we compare our results with naive least squares with respect to standard errors for the COVID-19 data (continuous outcomes) in \Cref{tab:se_cont} and with the generalized linear model for the financial aid data (binary outcome ) in \Cref{aid_fat_table}. Both the tables demonstrate that the standard errors of the network effect estimate $\beta$ are uniformly smaller for least squares and GLM as we expect, likely a representation of the overconfidence of the least squares model or GLM that comes from ignoring network induced correlation. This case demonstrates that using least squares would incorrectly dramatically increase both the magnitude and degree of significance of the network effects on incidence rates, particularly in the months post-lockdown when travel bans were in place. In contrast, the GEE model provides a more realistic view.

\vskip .1cm

\begin{table}[h]
\centering
\caption{\textbf{Estimates of the parameters for weighted and unweighted networks under the GEE and naive least squares model.} $\beta, \alpha_1, \alpha_2, \alpha_3$ and $\alpha_4$ correspond to the coefficients of the adjacency matrix (network effect), intercept, population, GDP, and percentage of the urban population, respectively. The estimated parameters are equal for both the GEE and naive least squares. The networks are constructed using country-to-country number of flights. The numbers in the parenthesis correspond to unweighted network obtained by thresholding the adjacency matrix  while the others correspond to weighted network obtained by scaling the number of flights by per million population of the recipient country.}
\label{tab1}
\resizebox{\textwidth}{!}{%
\begin{tabular}{c|c|c|c|c|c}
\hline
Month of 2020 &
  \begin{tabular}[c]{@{}c@{}}$\beta$\\ Weighted \\ (Unweighted)\end{tabular} &
  \begin{tabular}[c]{@{}c@{}}$\alpha_1$\\ Weighted \\ (Unweighted)\end{tabular} &
  \begin{tabular}[c]{@{}c@{}}$\alpha_2$\\ Weighted \\ (Unweighted)\end{tabular} &
  \begin{tabular}[c]{@{}c@{}}$\alpha_3$\\ Weighted \\ (Unweighted)\end{tabular} &
  \begin{tabular}[c]{@{}c@{}}$\alpha_4$\\ Weighted \\ (Unweighted)\end{tabular} \\ \hline
\multirow{8}{*}{\begin{tabular}[c]{@{}c@{}}Jan \\ \\ Feb \\ \\ Mar \\ \\ Apr \\ \end{tabular}} &
  \multirow{8}{*}{\begin{tabular}[c]{@{}c@{}}0.0038 \\ (1.8800)\\ 0.0014 \\ (2.1900)\\ 0.0030 \\ (2.7190)\\ 0.0328 \\ (4.6574)\end{tabular}} &
  \multirow{8}{*}{\begin{tabular}[c]{@{}c@{}}-0.0006 \\ (-0.0009)\\ -0.0186 \\ (-0.0239)\\ -3.7430 \\ (-3.9231)\\ -7.8275 \\ (-9.6532)\end{tabular}} &
  \multirow{8}{*}{\begin{tabular}[c]{@{}c@{}}0.0000 \\ (0.0000)\\ 0.0000 \\ (0.0000)\\ -0.0005 \\ (-0.0005)\\ -0.0123 \\ (-0.0084)\end{tabular}} &
  \multirow{8}{*}{\begin{tabular}[c]{@{}c@{}}0.0013 \\ (0.0013)\\ 0.0576 \\ (0.0556)\\ -0.0231 \\ (-0.1444)\\ 0.4685 \\(0.0626)\end{tabular}} &
  \multirow{8}{*}{\begin{tabular}[c]{@{}c@{}}0.0000 \\ (0.0000)\\ -0.0072 \\ (-0.0094)\\ 9.8862 \\ (9.8844)\\ 54.9768 \\ (62.3300)\end{tabular}} \\
 &  &  &  &  &  \\
 &  &  &  &  &  \\
 &  &  &  &  &  \\ 
 &  &  &  &  &  \\
 &  &  &  &  &  \\
 &  &  &  &  &  \\
 &  &  &  &  &  \\ \hline
\end{tabular}%
}
\end{table}


\begin{table}[]
\centering
\caption{\textbf{Comparison of the standard error between naive linear regression and GEE  for the weighted networks.} The numbers reported in the parenthesis correspond to our GEE method while the others correspond to least squares. $\beta, \alpha_1, \alpha_2, \alpha_3$ and $\alpha_4$ correspond to the coefficients of the adjacency matrix, intercept, population, GDP, and percentage of the urban population respectively. ***, **, *, ., etc. represent significance code corresponding to P-values: *** (P-value $<.001$), ** ($.001<$ P-value $<.01$), * ($.01<$ P-value $< .05$), . ($.05<$ P-value $<.1$), empty space (P-value $>.1$).}
\label{tab:se_cont}
\resizebox{\textwidth}{!}{%
\begin{tabular}{c|c|c|c|c|c}
\hline
Month of 2020 &
  \begin{tabular}[c]{@{}c@{}}$\beta$\\ SE: GEE (LS) \end{tabular} &
  \begin{tabular}[c]{@{}c@{}}$\alpha_1$\\ SE: GEE (LS)  \end{tabular} &
  \begin{tabular}[c]{@{}c@{}}$\alpha_2$\\ SE: GEE (LS)  \end{tabular} &
  \begin{tabular}[c]{@{}c@{}}$\alpha_3$\\ SE: GEE (LS)   \end{tabular} &
  \begin{tabular}[c]{@{}c@{}}$\alpha_4$\\ SE: GEE (LS)   \end{tabular} \\ \hline
Jan &
  \begin{tabular}[c]{@{}c@{}}2.73 $\times 10^{-2}$ \\ ($3.64 \times 10^{-3}$)\end{tabular} &
  \begin{tabular}[c]{@{}c@{}}$6.66 \times 10^{-4}$  \\ ($9.25 \times 10^{-4}$)\end{tabular} &
  \begin{tabular}[c]{@{}c@{}}$3.26 \times 10^{-6}$ \tiny{**} \\ ($3.58 \times 10^{-6}$) \tiny{***}\\ \end{tabular} &
  \begin{tabular}[c]{@{}c@{}}$4.7 \times 10^{-4}$ \\ ($5.00 \times 10^{-4}$)\end{tabular} &
  \begin{tabular}[c]{@{}c@{}}$1.21 \times 10^{-3}$ \\ ($1.27 \times 10^{-3}$)\end{tabular} \\ \hline
Feb &
  \begin{tabular}[c]{@{}c@{}}$3.36 \times 10^{-3}$ \\ ($3.23 \times 10^{-3}$)\end{tabular} &
  \begin{tabular}[c]{@{}c@{}}$2.94 \times 10^{-2}$ \\ ($3.77 \times 10^{-2}$)\end{tabular} &
  \begin{tabular}[c]{@{}c@{}}$1.47 \times 10^{-4}$ \tiny{**}\\ ($6.43 \times 10^{-5}$) \tiny{***}\\ \end{tabular} &
  \begin{tabular}[c]{@{}c@{}}$2.04 \times 10^{-2}$ \\ ($7.91 \times 10^{-3}$)\end{tabular} &
  \begin{tabular}[c]{@{}c@{}}$4.81 \times 10^{-2}$ \\  ($5.23 \times 10^{-2}$)\end{tabular} \\ \hline
Mar &
  \begin{tabular}[c]{@{}c@{}}$5.53 \times 10^{-4}$ \tiny{***} \\ ($2.34 \times 10^{-4}$)  \\ \end{tabular} &
  \begin{tabular}[c]{@{}c@{}}2.129425\\  (2.71420)\end{tabular} &
  \begin{tabular}[c]{@{}c@{}}$2.567 \times 10^{-3}$\\ ($4.79 \times 10^{-3}$)\end{tabular} &
  \begin{tabular}[c]{@{}c@{}}0.388691\\  (0.58001)\end{tabular} &
  \begin{tabular}[c]{@{}c@{}}3.703839 \tiny{*}\\ (3.74144) \tiny{*}  \\ \end{tabular} \\ \hline
Apr &
  \begin{tabular}[c]{@{}c@{}}$2.47 \times 10^{-3}$ \tiny{***} \\ ($7.09 \times 10^{-4}$) \tiny{***} \\ \end{tabular} &
  \begin{tabular}[c]{@{}c@{}}2.94923 \tiny{*}\\ (5.48968)   \\ \end{tabular} &
  \begin{tabular}[c]{@{}c@{}}$1.247 \times 10^{-2}$ \\ ($2.408 \times 10^{-2}$)\end{tabular} &
  \begin{tabular}[c]{@{}c@{}}0.03061  \\ (0.087040)\end{tabular} &
  \begin{tabular}[c]{@{}c@{}}14.48825  \tiny{***}\\ (23.69662) \tiny{*}  \\ \end{tabular} \\ \hline
\end{tabular}%
}
\end{table}

\newpage

\begin{table}
\centering
\caption{\textbf{Estimates of the parameters and their standard errors for weighted and unweighted networks under the GEE and GLM model for financial aid data.} The numbers reported in the parenthesis correspond to
our GEE method while the others correspond to GLM method. $\beta, \alpha_1, \alpha_2, \alpha_3$ and $\alpha_4$ correspond to the coefficients of the adjacency matrix (network effect), intercept, population, GDP, and percentage of the urban population, respectively. The estimated parameters are equal for both the GEE and GLM. **, ., etc. represent significance code corresponding to different P-values. ** ($.001<$ P-value $<.01$), * ($.01<$ P-value $< .05$), . ($.05<$ P-value $<.1$), empty space (P-value $>.1$).}
\label{aid_fat_table}
\begin{tabular}{c|cc|cc}
\hline
Coefficients &
  \multicolumn{2}{c|}{Weighted} &
  \multicolumn{2}{c}{Unweighted} \\ \hline
 &
  \multicolumn{1}{c|}{\begin{tabular}[c]{@{}c@{}}Estimates\\ GEE (GLM)\end{tabular}} &
  \multicolumn{1}{c|}{\begin{tabular}[c]{@{}c@{}}SE\\ GEE (GLM)\end{tabular}} &
  \multicolumn{1}{c|}{\begin{tabular}[c]{@{}c@{}}Estimate\\ GEE (GLM)\end{tabular}} &
  \multicolumn{1}{c}{\begin{tabular}[c]{@{}c@{}}SE\\ GEE (GLM)\end{tabular}} \\ \hline
\begin{tabular}[c]{@{}c@{}}$\beta$ \\ \\ $\alpha_1$ \\ \\ $\alpha_2$ \\ \\ $\alpha_3$ \\ \\ $\alpha_4$\end{tabular} &
  \multicolumn{1}{c|}{\begin{tabular}[c]{@{}c@{}}4.0614 \\(4.0614)\\ 0.1448 \\(0.1448)\\ 0.0194 \\(0.0194)\\ -4.0251 \\(-4.0251)\\ -2.1891 \\(-2.1891)\end{tabular}} &
  \multicolumn{1}{c|}{\begin{tabular}[c]{@{}c@{}}1.3524 \tiny{**} \\(1.2974) .\\ 0.8081 \\(1.4187)\\ 0.0120 \\(0.0131)\\ 3.2373 \\(2.8112) \\ 1.4444 \\ (2.4090)\end{tabular}} &
  \multicolumn{1}{c|}{\begin{tabular}[c]{@{}c@{}}9.4789 \\(9.4789)\\ -0.1094 \\(-0.1094)\\ 0.0167 \\(0.0167)\\ -3.8770 \\(-3.8770)\\ -0.5197 \\(-0.5197)\end{tabular}} &
  \multicolumn{1}{c}{\begin{tabular}[c]{@{}c@{}}1.5573 \tiny{*}\\(1.5212)\\ 0.0329 \\(0.0410)\\ 0.0111 \\(0.0128)\\ 2.9386 \\(3.0197)\\ 0.0634 \\(0.2043)\end{tabular}} \\ \hline
\end{tabular}%
\end{table}

\section{Discussion}\label{Discussion}

We have proposed a generalized estimating equation (GEE) approach to network regression model. Assuming independent community structure and  using the simultaneous estimation of memberships, the GEE approach allows for community dependent covariate coefficient estimation, and thereby provides a flexible and efficient solution to the network regression model. Moreover, this allows us to do a hypothesis testing of the network regression parameter $\beta$ which helps us to decide the importance of including such term in our analysis. We provided a relevant real data example of COVID-19 cases along with baseline covariates such as GDP, population size, and percentage of urban population of countries across different continents, and the number of commercial flights traveling between them to study the importance of travel bans to mitigate the spread of the COVID-19 pandemic. We have constructed the adjacency matrix from the count of flights rendered in a network of countries via a stochastic block model where each block contains countries having a similar number of flights. Our proposed model has helped us to understand the importance of the baseline covariates vs network effect. Since we have dealt with longitudinal data, it is also instructive to note that our proposed model offers us the flexibility  of clustering both in the network space and also over time. 

One limitation of our results is the balanced community assumption made in Proposition \ref{const_p_q}, which will not hold for all networks. This assumption can be relaxed somewhat with small departures from the balanced design. Further, under the stochastic block model, extremely unbalanced networks still allow for consistent estimation of $p$ and $q$, however for a more flexible network model that allows for community-specific edge probabilities, consistent estimation of edge probabilities requires each community to grow asymptotically with $n$.

\section{Appendix}\label{Appendix}

\subsection{Proof of \Cref{const_p_q}}\label{proof_const_p_q}
From the construction of the adjacency matrices, one can note that the entries $A_{ij}$s are i.i.d. Bernoulli random variables with
$$
E(A_{ij})=
\begin{cases}
p, \;\;\ \text{if } i, j \text{ belong to the same community}\\
q, \;\;\ \text{otherwise} \\
\end{cases}
$$
and 
$$
Var(A_{ij})=
\begin{cases}
p(1-p), \;\;\ \text{if } i, j \text{ belong to the same community}\\
q(1-q), \;\;\ \text{otherwise} \\
\end{cases}
$$
Denote by $S$ the set $\{i,j: 1\leq i<j\leq n, i,j \text{ belongs to the same community}\}$, and its complement by $S'$. Therefore, for a directed graph, one can write

\begin{align*}
& E\Big(\sum_{i,j\in S}A_{ij}\Big)=2K{m \choose 2}p, \\& Var\Big(\sum_{i,j\in S}A_{ij}\Big)=2K{m \choose 2}p(1-p)=s_{m_p}^2,\\
& E\Big(\sum_{i,j\in S'}A_{ij}\Big)=K
(K-1)m^2q, \text{ and }\\
&Var\Big(\sum_{i,j\in S'}A_{ij}\Big)=K
(K-1)m^2q(1-q)=s_{m_q}^2.
\end{align*}
Next, one can verify Lindeberg condition to establish the central limit theorem:

\begin{align}\label{lindeberg_cond}
& \frac{\sum_{i,j\in S}A_{ij}-2K{m \choose 2}p}{\sqrt{2K{m \choose 2}p(1-p)}} \;\;\ \text{ and }\;\; \frac{\sum_{i,j\in S'}A_{ij}-K(K-1)m^2q}{\sqrt{K(K-1)m^2q(1-q)}}\overset{d}{\rightarrow} N(0,1)
\end{align}
as $Km^2p(1-p)$ and $Km^2q(1-q) \rightarrow \infty$.\\

The Lindeberg's condition requires us to verify

$$\frac{1}{s_{m_p}^2}\sum_{i,j\in S}E\Big(B_{ij}^2I_{\{|B_{ij}|\geq \epsilon s_{m_p}\}}\Big) \rightarrow 0 \;\;\ \text{ and }\frac{1}{s_{m_q}^2}\sum_{i,j\in S'}E\Big(B_{ij}^2I_{\{|B_{ij}|\geq \epsilon s_{m_q}\}}\Big) \rightarrow 0,$$
where $B_{ij}=A_{ij}-E(A_{ij})$. Since $|B_{ij}|\leq 1$, the above condition is satisfied when $Km^2p(1-p)$ and $Km^2q(1-q)$ tend to $\infty$ under which $\epsilon s_{m_p}$ and $\epsilon s_{m_q}$ are both greater than 1.\\ 

Dividing the numerator and denominator of \eqref{lindeberg_cond} by $2K{m \choose 2}$ and $K(K-1)m^2$ respectively, one obtains both

$m^{\gamma/2}\sqrt{2K{m \choose 2}}\frac{\sum_{i,j\in S}A_{ij}/\Big(2K{m \choose 2}\Big)-p}{\sqrt{m^{\gamma}p(1-p)}} \;\;\ \text{ and }\;\; m^{\gamma/2}\sqrt{K(K-1)m^2}\frac{\sum_{i,j\in S'}A_{ij}/\Big(K(K-1)m^2\Big)-q}{\sqrt{m^{\gamma}q(1-q)}}$ both converge to $N(0,1)$ distribution.
Since $\hat{p}=\sum_{i,j\in S}A_{ij}/\Big(K{m \choose 2}\Big)$, and $\hat{q}=\sum_{i,j\in S'}A_{ij}/\Big(K(K-1)m^2\Big)$,  $\hat{p}$ and $\hat{q}$ are $m^{1+\gamma/2}K^{1/2}$ and $m^{1+\gamma/2}K$ consistent respectively.


\subsection{Proof of \Cref{asym_norm}}\label{asym_norm_pf}
Here we present the sketch of the proof of \Cref{asym_norm}. One can write \eqref{gee_objective_fn} as
\begin{align}\label{obj_fun}
& \sum_{k=1}^{K}U_k(\alpha,\bm \beta,p,q)=\sum_{k=1}^{K}\bm D_k^{\top}\bm V_k^{-1}\bm S_k=0,
\end{align}
where $\bm S_k=\bm y_k-\bm \mu_k$, and $U$ is a function of the model parameters.

Let $\bm b=(\beta,\bm \alpha)^{\top}$ denote the vector of model parameters, and $\bm \pi =(p,q)^{\top}$ denote the vector of model parameters, and vector of within and between edge probabilities respectively. Letting $\bm b$ fixed, the Taylor expansion yields
\begin{align}\label{taylor}
& \frac{\sum_{k=1}^{K}U_k(\bm b, \bm \pi^*)}{K^{1/2}m^{1+\gamma/2}}=\frac{\sum_{k=1}^{K}U_k(\bm b, \bm \pi)}{K^{1/2}m^{1+\gamma/2}}+\frac{\sum_{k=1}^{K}\frac{\partial U_k(\bm b,\bm \pi)}{\partial \pi}}{K^{1/2}m^{1+\gamma/2}}m^{1+\gamma/2}(\bm \pi^*-\bm \pi)+o_P(1)\\ \nonumber
&\hspace{2.8cm} =\tilde{A}+\tilde{B}\tilde{C}+o_P(1).
\end{align}

One can note that $\tilde{B}=o_P(1)$ as $\partial U_k(\bm b,\bm \pi)/\partial \bm \pi$  are linear functions of $\bm S_k$'s defined in \eqref{obj_fun} whose means are zero, and $\tilde{C}=O_P(1)$ thanks to \Cref{const_p_q}. Therefore, $\frac{\sum_{k=1}^{K}U_k(\bm b, \bm \pi^*)}{K^{1/2}m^{1+\gamma/2}}$ is asymptotically equivalent to $\frac{\sum_{k=1}^{K}U_k(\bm b, \bm \pi)}{K^{1/2}m^{1+\gamma/2}}$ whose asymptotic distribution is multivariate normal with zero mean and covariance is equal to $V$ as defined in \Cref{asym_norm}. The proof is thus complete following \cite{liang1986longitudinal}.

\section{Acknowledgement}
RG would like to thank Anupam Kundu, postdoctoral associate at Yale School of Public Health and Thien Le, postdoctoral fellow at Harvard for many useful discussions regarding real data processing. JPO acknowledges support from R01AI138901, and IB acknowledges support from R01MH116884.

\medskip
\bibliographystyle{apacite}
\bibliography{reference}
\end{document}